\documentclass[review]{elsarticle}

\usepackage{lineno}
\usepackage{graphicx}
\usepackage{bm}

\usepackage{algorithm}
\usepackage{float}
\usepackage[noend]{algpseudocode}
\usepackage{xcolor}
\usepackage{amsthm}
\usepackage{amssymb}
\usepackage{amsmath}
\usepackage{float} 
\usepackage{array}
\usepackage{enumitem}
\usepackage{tikz}
\usepackage{pgflibraryarrows}
\usepackage{pgflibrarysnakes}

\modulolinenumbers[5]

\DeclareMathOperator{\E}{\mathbb{E}}

\newtheorem{definition}{Definition}

\newtheorem{lemma}{Lemma}

\newtheorem{corollary}{Corollary}

\newtheorem{theorem}{Theorem}

\theoremstyle{remark}
\newtheorem{remark}{Remark}

\newcommand{\R}{\mathbb{R}}
\newcommand*\dd{\mathop{}\!\mathrm{d}}

\algnewcommand\algorithmicforeach{\textbf{for each}}
\algdef{S}[FOR]{ForEach}[1]{\algorithmicforeach\ #1\ \algorithmicdo}

\journal{Journal Name}

\begin{document}

\begin{frontmatter}

\title{Expected Size of Random Tukey Layers and Convex Layers}

\author[1]{Zhengyang Guo\corref{cor1}}
\ead{GUOZ0015@e.ntu.edu.sg}

\author[1]{Yi Li}
\ead{yili@ntu.edu.sg}

\author[2]{Shaoyu Pei}
\ead{Shaoyupei@mail.fresnostate.edu}

\address[1]{School of Physical and Mathematical Sciences, Nanyang Technological University, Singapore}
\address[2]{College of Science and Mathematics,
           California State University, Fresno, California, United States}
           
\cortext[cor1]{Corresponding author}

% !TeX root = main.tex
\begin{abstract}
We study the Tukey layers and convex layers of a planar point set, which consists of $n$ points independently and uniformly sampled from a convex polygon with $k$ vertices. We show that the expected number of vertices on the first $t$ Tukey layers is $O\left(kt\log(n/k)\right)$ and the expected number of vertices on the first $t$ convex layers is $O\left(kt^{3}\log(n/(kt^2))\right)$. We also show a lower bound of $\Omega(t\log n)$ for both quantities in the special cases where $k=3,4$. The implications of those results in the average-case analysis of two computational geometry algorithms are then discussed.
\end{abstract}


\end{frontmatter}

%\linenumbers

% !TeX root = main.tex
\section{Introduction}
The motivation of this work is to understand the combinatorial and geometric properties of random convex layers and Tukey layers of planar point sets $X$. The convex layers of $X$ are a sequence of nested convex polygons whose vertices form a partition of $X$. The Tukey layers are the cells of a partition of $X$, in which each cell consists of all points in $X$ of the same Tukey depth~\cite{tukey1975mathematics}. We refer the readers to Definitions~\ref{Convex Layer Structure} and~\ref{def::t-shell} for precise definitions. Each Tukey layer, as we shall prove in Lemma~\ref{lem:convex position}, is exactly the vertices of a convex polygon. % In this paper, the underlying point set $X$ is assumed to be uniformly sampled from a convex polygon.

There has been a long research history on the expected size of the convex hull of a random point set \cite{RS63,dwyer1988average,hueter1994convex,har2011expected}, the relation between the expected size and the expected area of the convex hull \cite{affentranger1991convex,dwyer1988convex}, and the expected convex depth \cite{dalal2004counting}. However, few of them study convex layers. In fact, the vertices on the first $t$ convex layers, denoted by $V_{[t]}(X)$, are closely related to the partial enclosing problem introduced by Atanassov et al.\ in \cite{atanassov2009algorithms}. The objective of this problem is to find the convex hull with the minimum area that encloses $(n-t)$ of the $n$ points in $X$. The $t$ excluded points are regarded as outliers, as in many works that study the partial covering, for example \cite{guha2019distributed}, \cite{gupta2017local} and \cite{har2004shape},.

% This implies that there might be situations where the algorithm is more efficient than in the others. %it from a probabilistic point of view (the average time complexity).

In \cite{atanassov2009algorithms}, Atanassov et al.\ give an algorithm with the worst-case time complexity of $O\left(n\log{n}+\binom{4t}{2t}(3t)^tn\right)$, where the $n$ in the second term $\binom{4t}{2t}(3t)^tn$ refers to the size $\left|V_{[t]}(X)\right|$ in the worst case. However, the actual runtime seldom meets such worst cases. To give an overall measure on the efficiency of the algorithm, it makes more sense to study the average time complexity. Assuming that $X$ is uniformly sampled from a convex $k$-gon as in \cite{RS63,efron1965convex,dwyer1988convex,affentranger1991convex,barany1999sylvester,har2011expected,buchta2012boundary}, we shall prove in Section~\ref{section: upper bound for convex layer} that $\E\left|V_{[t]}(X)\right|=O(kt^{3}\log(n/(kt^2)))$, which is $o(n)$ when $t = o((n/(k\log(nk))^{1/3})$. As a consequence, the expected complexity of Atanassov et al.'s algorithm in \cite{atanassov2009algorithms} is $O(n\log{n}+\binom{4t}{2t}(3t)^{t}k t^3\log(n/t^2))$. This explains the gap between the worst case complexity and the actual runtime.

In addition, we study the expected number of vertices on the first $t$ Tukey layers $U_{[t]}(X)$ as defined in Definition~\ref{def::t-shell}. This is also related to a partial shape fitting problem \cite{guo2020minimum} in which the parallelogram rather than the convex polygon as in \cite{atanassov2009algorithms} is concerned. The time complexity of the algorithm in \cite{guo2020minimum} is $O(n^2t^4+n^2\log{n})$, where the $n$ in the first term $n^2t^4$ refers to $\left|U_{[t]}(X)\right|$ in the worst case. As we shall prove $\E\left|U_{[t]}(X)\right|=O(kt\log(n/k))$ in Section~\ref{section: upper bound for Tukey layer}, the expected time complexity is then $O\left(kt^5n\log(n/k)+n^2\log{n}\right)$, smaller than the worst-case complexity when $\Omega\left((n/k)^{1/5}\right)\leq t\leq O(n/(k\log{n}))$.

It is beneficial to study the convex hulls and Tukey layers together. Their close relation is shown in Lemma~\ref{lem:convex relative position} that $U_{[t]}(X)\subseteq V_{[t]}(X)$. An upper bound on $\left|V_{[t]}(X)\right|$ is then automatically an upper bound on $\left|U_{[t]}(X)\right|$ and a lower bound on $\left|U_{[t]}(X)\right|$ is automatically a lower bound on $\left|V_{[t]}(X)\right|$.

% !TeX root = main.tex
\subsection{Notation and Definitions}
\label{subsec:notatiions}
We introduce the notation and definitions before reviewing the existing works. Let $X$ be a planar point set and $n=|X|$ be its size. When $X$ is a random point set, we use $\mathcal{P}$ to denote the convex polygon from which $X$ is sampled, and $k$ to denote the number of vertices of $\mathcal{P}$. Throughout this work, the convex polygon $\mathcal{P}$ is always closed and, without loss of generality, we assume the area of $\mathcal{P}$ is $1$. We now present the definition of the convex layer structure as in \cite{chazelle1985convex}.
\begin{definition}[Convex Layer]\label{Convex Layer Structure}
Given a planar point set $X$, the first convex layer $H_1(X)$ is defined to be the convex hull $H(X)$ of the whole point set. The $t$-th convex layer $H_t(X)$ is inductively defined to be the convex hull of the remaining points, after the points on the first $(t-1)$ convex layers have been removed from $X$.
\end{definition}

\begin{definition}[Convex Depth] The convex depth of $p\in X$ is said to be $t$ if $p$ is a vertex of $H_t(X)$.
\end{definition}

Next we define the Tukey layers, for which we need to introduce a classical notion known as the Tukey depth~\cite{tukey1975mathematics}. Instead of using the original definition, we use the following equivalent form for finite point sets.

\begin{definition}[Tukey Depth]
\label{def::Tukey depth}
Given a set $X$ of planar points, the Tukey depth of a point $p\in X$ is defined to be $N(p)+1$, where $N(p)$ is the minimum number of points in $X$ that are contained in any open half-plane with p on its boundary.
\end{definition}

%\begin{definition}[Shell Order]Suppose $p\in X$. For a line $\ell$ through $p$, we let $N^{(1)}(p, \ell, X)$ and $N^{(2)}(p, \ell, X)$ denote the number of points of $X$ on its two sides respectively, and $N(p, \ell, X):=\min\left(N^{(1)}(p, \ell, X), N^{(2)}(p, \ell, X)\right)$. The shell order of $p$ is then defined to be $\displaystyle\min_{\ell: p\in\ell}N(p, \ell, X)$ where $\{\ell: p\in\ell\}$ consists of all lines passing through the point $p$.\end{definition}

\begin{remark}
	For brevity, we use ``one side of a line $\ell$'' to refer to one of the two open half-planes induced by $\ell$. Hence, if a point $p$ is on one side of a line $\ell$, the point $p$ is in an open half-plane induced by $\ell$. Besides, when we say a point is above (below) a line, we do not include the line either.
\end{remark}

\begin{remark}
Intuitively, if a point $p$ has Tukey depth $t$, then for all lines $\ell$ through $p$, there cannot be fewer than $(t-1)$ points on either side of $\ell$. At the same time, there exists a line $\ell_{0}$ through $p$ such that there are exactly $(t-1)$ points on one side of $\ell_{0}$.
\end{remark}

\begin{definition}[Tukey Layer]
\label{def::t-shell}
For $t\geq 1$, the subset $U_t(X)$ of $X$ is defined to be the set of points of Tukey depth $t$. The $t$-th Tukey layer, denoted by $S_t(X)$, is defined to be convex hull of $U_t(X)$. The size of $S_t(X)$ is defined to be $\left|U_t(X)\right|$.
\end{definition}

An illustration of Tukey layers is shown in Figure~\ref{fig:convex position}. As we shall prove in Lemma~\ref{lem:convex position}, the points in $U_t(X)$ are in the convex position and are thus exactly the vertices of $S_t(X)$, hence our definition of the size of $S_t(X)$ makes sense. The frequently used notations are listed in Table~\ref{tab:notations}. Note that $S_1(X)=H(X)$ by definition. For convenience, we also let $V_{[t]}(X):=\bigcup_{i=1}^{t}V_{i}(X)$ and $U_{[t]}(X):=\bigcup_{i=1}^{t}U_{i}(X)$.

\begin{table}[!t]
\centering
\begin{tabular}{|c|l|c|l|}
\hline
Symbol & Definition & Symbol & Definition\\
\hline
$H(X)$ & the convex hull of $X$ & $H_{t}(X)$ & the $t$-th convex layer of $X$\\
\hline
$V(X)$ & the vertices of $H(X)$ & $V_{t}(X)$ & the vertices of $H_{t}(X)$\\
\hline
& & $S_{t}(X)$ & the $t$-th Tukey layer of $X$\\
\hline
& & $U_{t}(X)$ & the vertices of $S_{t}(X)$\\
\hline
$A(X)$ & the area of $H_1(X)$ or $S_1(X)$ & $A_{t}(X)$ & the area of $S_{t}(X)$\\
\hline
\end{tabular}
\caption{Notations used in this work.}\label{tab:notations} 
\end{table}

% !TeX root = main.tex
\subsection{Related Work}
The main results in Section~\ref{section: upper bound for Tukey layer} and~\ref{section: upper bound for convex layer} are proved using the techniques developed for computing the expected convex hull size. We thus review the works that study the random convex hull, in terms of its area and the number of its vertices. Most of the research interests have been in their expectations, concentration bounds and asymptotic behaviors.

A fundamental result is that the expected size of a random convex hull is $O(k\log{n})$, when a large number $n$ of points are independently and uniformly sampled from a convex $k$-gon. The result was first stated by R\'eyi and Sulanke in~\cite{RS63} and a geometric proof was later provided by Har-Peled~\cite[Section 2]{har2011expected}. By the relation $\E\left|V(X)\right| = n\left[1-\E A(X)\right]$ proposed in \cite{efron1965convex} (the area of the $k$-gon is assumed to be 1 without loss of generality), an upper bound on $\E\left|V(X)\right|$ will follow from a lower bound on $\E A(X)$. Thus in \cite{har2011expected}, the effort is devoted to deriving a lower bound on the expected area of the convex hull. A critical observation in \cite[Section 2]{har2011expected} is that, if $p\in X$ is a vertex of the convex hull, then there exists a line $\ell$ through $p$ such that one side of $\ell$ contains no points of $X$. This gives a necessary condition on $p\in H(X)$, and a lower bound on the probability of the event $p\in H(X)$ can then be obtained. Multiplying this lower bound by $n$ immediately yields an lower bound on $\E A(X)$.

In addition, there have been a number of studies on the asymptotic behaviours of the convex hull size, such as \cite{RS63,affentranger1991convex,barany1988convex,masse2000lln,schutt1991convex}. R\'enyi and Sulanke proved that, given $X$ uniformly sampled from a convex $k$-gon on a plane, the expected size of the convex hull $\E\left|V(X)\right|$ is asymptotically $\frac{2}{3}k\log n + O(1)$ as $n\to\infty$, where the constant term depends on the polygon~\cite{RS63}. Affentranger and Wieacker generalized the result to higher dimensions and showed that, given that $X$ is uniformly sampled from a simple polytope in $\mathbb{R}^{d}$ with $k$ vertices, $\E\left|V(X)\right| = \frac{d}{(d+1)^{d-1}}k\log^{d-1}{n} + O(\log^{d-2}n)$~\cite{affentranger1991convex}. Masse proved that in the planar case, $|V(X)|/(\frac{2}{3}k\log n)$ converges to $1$ in probability~\cite{masse2000lln}.

There are also studies that assume different underlying distribution for the point set. When the $n$ points are sampled independently from a coordinate-wise independent distribution in $\mathbb{R}^{d}$, it is proved by He et al.\ in \cite{he2018maximal} that the expected size of the $t$-th convex layer is  $O(t^{d}\log^{d-1}(n/t^d))$. Some studies assume the point set is sampled independently and uniformly from other shapes rather than a convex polygon. In the case of a disc, the expected size of the convex hull is $\Theta(n^{1/3})$, due to Raynaud~\cite{raynaud1970enveloppe}.

% !TeX root = main.tex
\subsection{Our Contribution}
In this work, we introduce a new definition called Tukey layer and provide some fundamental properties of it. Then we study the expected size of the Tukey layers and convex layers when the point set $X$ is uniformly sampled from a $k$-gon. We show that the expected number of vertices of the first $t$ Tukey layers $\E\left|U_{[t]}(X)\right|=O(kt\log(n/k))$ and that of the first $t$ convex layers $\E\left|V_{[t]}(X)\right| = O(kt^3\log(n/kt^2))$. The first work to study the expected size of convex layers is \cite{he2018maximal} where He et al. proved that $\E\left|V_{t}(X)\right| = O(t^2\log(n/t^2))$ when $X$ follows a continuous component independent distribution. Their result can be extended to the cases when $X$ is sampled from a square or more generally a parallelogram, and their bound $O(t^2\log(n/t^2))$ is better than ours $O(t^3\log(n/t^2))$ in such cases. On the other hand, the techniques developed in \cite{he2018maximal} are towards the continuous component independent distribution, and we find it hard to extent them to other polygonal shapes except square or parallelogram. We also prove a matching lower bound $\E\left|U_{[t]}(X)\right|=\Omega (t\log{n})$ when $X$ is sampled from a triangle or a parallelogram, which, since $U_{[t]}(X) \subseteq V_{[t]}(X)$, is also a lower bound for $\E\left|V_{[t]}(X)\right|$ in the two special cases. Finally, we show that the two upper bounds are helpful in understanding the average case complexity of two partial shape fitting algorithms, both of which aim to enclose $(n-t)$ of the $n$ given points with a shape of the minimum area. One shape is parallelogram and the other is convex polygon.

% !TeX root = main.tex
\subsection{Organization}
In Section \ref{Preliminaries} we give the fundamental properties of convex layers and Tukey layers. In Section \ref{section: upper bound for Tukey layer}, we present the proof of the upper bound on the expected size of the first $t$ Tukey layers, when the $n$ points in $X$ are sampled from a convex polygon. In Section~\ref{section: upper bound for convex layer}, we prove the upper bound on the expected size of the first $t$ convex layers under the same setting. In Section~\ref{section: low bound for convex shell}, we derive the lower bounds on the expected size of the first $t$ Tukey layers for two special cases. Finally in Section \ref{section: application}, we apply our results to the average-case analysis of two shape fitting algorithms.

% !TeX root = main.tex
\section{Preliminaries}
\label {Preliminaries}
In this section, we prepare some fundamental facts on Tukey Layers and convex layers. The readers are recommended to have a look through the statements to get familiar with these properties. Nonetheless, we include the proofs for completeness. To the best of our knowledge, the observations on Tukey layers are new and not found in the literature. 

\subsection{Convex Layers, Tukey Layers and their Relation}
The following lemma shows that the points in $U_t(X)$ are exactly the vertices of the $t$-th Tukey layer $S_t(X)$, which justifies referring the size of $S_t(X)$ to $|U_t(X)|$ as we mentioned after Definition~\ref{def::t-shell}.

\begin{lemma}
\label{lem:convex position}
For a planar point set $X$, the points in the $t$-th Tukey layer of $X$ are in the position of a convex polygon. Equivalently, $U_{t}(X)$ has only one convex layer. 
\end{lemma}

\begin{proof}
Suppose there are at least two convex layers in $U_{t}(X)$. Let $V_{1}$ denote the vertices of the convex hull of $U_{t}(X)$, and $V_{2}:=U_{t}(X)\setminus V_{1}$. For any point $p\in V_{2}$, let $\ell$ be the line through $p$ such that there are exactly $(t-1)$ points on one side. Notice that $\ell$ is through $p$ and thus also through the interior of the convex hull of $U_{1}(X)$. Hence, on the side of $\ell$ that contains $(t-1)$ points, there must exist a point $q$ which belongs to $V_{1}$. This implies that for the line $\ell^{\prime}$ through $q$ and parallel to $\ell$, there are at most $(t-2)$ points on its one side. This contradicts the fact that $q\in V_{1}\subseteq U_{t}(X)$. Finally we conclude that there can be only one single convex layer in each $U_{t}(X)$.
\end{proof}

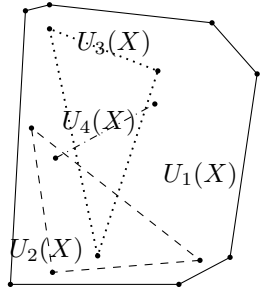
\begin{figure}[!t]
\centering
\begin{minipage}[h]{0.9\textwidth}
\centering
\begin{tikzpicture}
\node at (0.28, 3.76) {};
\draw node[fill,circle,minimum width=2pt, inner sep=0pt] at (0.28, 3.76){};

\node at (0.6, 3.84) {};
\draw node[fill,circle,minimum width=2pt, inner sep=0pt] at (0.6, 3.84){};

\node at (0.6, 3.52) {};
\draw node[fill,circle,minimum width=2pt, inner sep=0pt] at (0.6, 3.52){};

\node at (2.76, 3.6) {};
\draw node[fill,circle,minimum width=2pt, inner sep=0pt] at (2.76, 3.6){};

\node at (3.36, 2.92) {};
\draw node[fill,circle,minimum width=2pt, inner sep=0pt] at (3.36, 2.92){};

\node at (3, 0.48) {};
\draw node[fill,circle,minimum width=2pt, inner sep=0pt] at (3, 0.48){};

\node at (2.32, 0.12) {};
\draw node[fill,circle,minimum width=2pt, inner sep=0pt] at (2.32, 0.12){};

\node at (2.6, 0.44) {};
\draw node[fill,circle,minimum width=2pt, inner sep=0pt] at (2.6, 0.44){};

\node at (0.08, 0.12) {};
\draw node[fill,circle,minimum width=2pt, inner sep=0pt] at (0.08, 0.12){};

\node at (0.64, 0.28) {};
\draw node[fill,circle,minimum width=2pt, inner sep=0pt] at (0.64, 0.28){};

\node at (1.24, 0.5) {};
\draw node[fill,circle,minimum width=2pt, inner sep=0pt] at (1.24, 0.5){};

\node at (0.36, 2.2) {};
\draw node[fill,circle,minimum width=2pt, inner sep=0pt] at (0.36, 2.2){};

\node at (0.68, 1.8) {};
\draw node[fill,circle,minimum width=2pt, inner sep=0pt] at (0.68, 1.8){};

\node at (2.04, 2.96) {};
\draw node[fill,circle,minimum width=2pt, inner sep=0pt] at (2.04, 2.96){};

\node at (2, 2.52) {};
\draw node[fill,circle,minimum width=2pt, inner sep=0pt] at (2, 2.52){};

\draw (0.28, 3.76) -- (0.6, 3.84) -- (2.76, 3.6) -- (3.36, 2.92) -- (3, 0.48) -- (2.32, 0.12) -- (0.08, 0.12) -- (0.28, 3.76);

\draw[dashed] (2.6, 0.44) -- (0.64, 0.28) -- (0.36, 2.2) -- (2.6, 0.44);

\draw[dotted, line width=0.25mm] (0.6, 3.52) -- (1.24, 0.5) -- (2.04, 2.96) -- (0.6, 3.52);

\draw[dash dot] (0.68, 1.8) -- (2, 2.52);

\node at (3.35,2) [label=below left:$U_1(X)$]{};

\node at (1.3,1) [label=below left:$U_2(X)$]{};

\node at (2.2,3.75) [label=below left:$U_3(X)$]{};

\node at (2,2.7) [label=below left:$U_4(X)$]{};
\end{tikzpicture}
\caption{The boundary of the first three Tukey layers $U_1(X)$, $U_2(X)$ and $U_3(X)$ is plotted in solid, dashed, and dotted lines, respectively. The fourth Tukey layer $U_4(X)$ degenerates to a line segment, plotted in dashed dots. The vertices in each Tukey layer are in the convex positions.}
\label{fig:convex position}
\end{minipage}
\end{figure}

The next lemma relates Tukey layers and convex layers.

\begin{lemma}
\label{lem:convex relative position}
It holds that $U_{[t]}(X)\subseteq V_{[t]}(X)$.
\end{lemma}

\begin{proof}
If a point $p\in X\setminus V_{[t]}(X)$, then $p$ can only lie on the $(t+1)$-st or a deeper layer of $X$. On any side of any line passing through $p$, there must be at least one vertex from each previous layer, including the $1$-st to the $t$-th. In total there are at least $t$ points and by Definition~\ref{def::t-shell} it holds that $p\notin U_{[t]}(X)$. In conclusion, $U_{[t]}(X)\cap (X \setminus V_{[t]}(X))=\emptyset$ and thus $U_{[t]}(X)\subseteq V_{[t]}(X)$.
\end{proof}

The following lemma discusses the relative position of Tukey layers. It shows that the vertices on the first $t$ Tukey layers are outside the $(t+1)$-st Tukey layer.
\begin{lemma}
\label{lem:shell position}
It holds that $U_{[t]}(X)\cap S_{t+1}(X)=\emptyset$. As a consequence, $S_{t}(X)\subseteq H\left(X\setminus U_{[t-1]}(X)\right)$.
\end{lemma}

\begin{proof}
Suppose not. We let $p\in U_{[t]}(X)\cap S_{t+1}(X)$ and $\ell$ be a line through $p$, on one side of which there are at most $(t-1)$ points.

If $\ell$ intersects the interior of $S_{t+1}(X)$, then there must be a $q\in U_{t+1}(X)$ on the side of $\ell$ where there are at most $(t-1)$ points. Let $\ell^{\prime}$ denote the line through $q$ and parallel to $\ell$. Then there are at most $(t-2)$ points on one side of $\ell^{\prime}$ and this contradicts the fact that $q\in U_{t+1}(X)$.

If $\ell$ does not intersect the interior of $S_{t+1}(X)$, then $p$ must lie on a side $rq$ of the boundary of $S_{t+1}(X)$. Here $r$, $q\in U_{t+1}(X)$ and the line segment $rq$ must be on the line $\ell$. As there are at most $(t-1)$ points on one side of $\ell$, we then have $r$, $q\in U_{[t]}(X)$, contradictory to the assumption that $rq$ is a side of the boundary of $S_{t+1}(X)$.
\end{proof}

\begin{lemma}
\label{lem:shell divide}
If $X_{1}\cup X_{2}=X$, then $U_{[t]}(X)\subseteq U_{[t]}(X_{1})\cup U_{[t]}(X_{2})$.
\end{lemma}

\begin{proof}
For each point $p\in U_{[t]}(X)$, there exists a line $\ell$ through it, on one side of which there are at most $(t-1)$ points of $X$. Then there will be neither more than $(t-1)$ points of $X_{1}$ nor more than $(t-1)$ points of $X_{2}$ on the same side of $\ell$. Then we have $p\in U_{[t]}(X_{1})$ when $p\in X_{1}$, and $p\in U_{[t]}(X_{2})$ when $p\in X_{2}$.
\end{proof}
The following corollary is a generalization to $k$ subsets.
\begin{corollary}
\label{coro:shell t-part}
Given $X=X_{1}\cup X_{2}\cup\cdots\cup X_{k}$, we have
\[U_{[t]}(X)\subseteq U_{[t]}(X_{1})\cup U_{[t]}(X_{2})\cup\cdots\cup U_{[t]}(X_{k}).\]
\end{corollary}

The following lemma is an analogous result of Lemma~\ref{lem:shell divide} for $V_{[t]}$.

\begin{lemma}
\label{lem:layer divide}
If $X_{1} \cup X_{2} = X$, then $H_{t}(X_1)\cup H_{t}(X_{2})\subseteq H_{t}(X)$ and $V_{[t]}(X)\subseteq V_{[t]}(X_1)\cup V_{[t]}(X_2)$.
\end{lemma}
\begin{proof}
We prove the lemma by induction on $t$. The statement is well-known when $t=1$. Assume it holds for $t$ and we shall prove it for $(t+1)$. By the induction hypothesis, $V_{[t]}(X)\subseteq V_{[t]}(X_1)\cup V_{[t]}(X\setminus X_1)$, we then have
\[X\setminus V_{[t]}(X)\supseteq X_{1}\setminus V_{[t]}(X)\supseteq X_{1}\setminus \left(V_{[t]}(X_1)\cup V_{[t]}(X\setminus X_1)\right)=X_1\setminus V_{[t]}(X_1).\]
Further by Definition~\ref{Convex Layer Structure}, \[H_{t+1}(X_1)=H\left(X_1\setminus V_{[t]}(X_1)\right) \subseteq H\left(X\setminus V_{[t]}(X)\right)=H_{t+1}(X).\]
Similarly, $H_{t+1}(X_2) \subseteq H_{t+1}(X)$. Therefore $H_{t+1}(X_1)\cup H_{t+1}(X_{2})\subseteq H_{t+1}(X)$.

Now we prove $V_{[t+1]}(X)\subseteq V_{[t+1]}(X_1)\cup V_{[t+1]}(X_2)$. For a point $p\in V_{[t+1]}(X)$, $p$ cannot be in the interior of $H_{t+1}(X)$. We have already shown that $H_{t+1}(X_1)\cup H_{t+1}(X_{2})\subseteq H_{t+1}(X)$, so $p$ cannot be in the interior of either $H_{t+1}(X_1)$ or $H_{t+1}(X_2)$. If $p\in X_1$, then $p\in V_{[t+1]}(X_1)$; otherwise $p\in V_{[t+1]}(X_2)$.
\end{proof}

\begin{corollary}
\label{coro:layer t-part}
Given $X=X_{1}\cup X_{2}\cup\cdots\cup X_{k}$, we have
\[V_{[t]}(X)\subseteq V_{[t]}(X_{1})\cup V_{[t]}(X_{2})\cup\cdots\cup V_{[t]}(X_{k}).\]
\end{corollary}

\subsection{Convex Depth}
The following lemma examines how the convex depth of a point $p$ in $X$ changes after an additional point $q$ in added to  $X$.
%\begin{lemma}Given a planar point set $X$ and a point $p\in X$, after an additional point $q$ is added into $X$, the convex order of $p$ can not decrease.\end{lemma}\begin{proof}Let $t$ denote the convex order of $p$ in $X\cup\{q\}$. In another word, $p\in V_{t}\left(X\cup\{q\}\right)$. By Lemma~\ref{lem:layer divide}, $V_{[t]}\left(X\cup\{q\}\right)\subseteq V_{[t]}(X)\cup V_{[t]}(\{q\})= V_{[t]}(X)\cup \{q\}$.This indicates that $p\in V_{t}\left(X\cup\{q\}\right)\subseteq V_{[t]}(X)\cup \{q\}$. As $p\neq q$, we have $p\in V_{[t]}(X)$ and the convex order of $p$ in $X$ can not be more than $t$.\end{proof}

\begin{lemma}
\label{lem:dahal}
Given a planar point set $X$ and a point $p\in X$, the convex depth of $p$ will either remain unchanged or increase at most by 1 after an additional point $q$ is added into $X$.
\end{lemma}

\begin{proof}
By the proof of \cite[Lemma 3.1]{dalal2004counting}, we know that $V_{t}(X)\subseteq V_{t}\left(X\cup\{q\}\right)\cup V_{t+1}\left(X\cup\{q\}\right)$. For $p\in V_{t}(X)$, either $p\in V_{t}\left(X\cup\{q\}\right)$ or $p\in V_{t+1}\left(X\cup\{q\}\right)$. In other words, the convex depth of $p$ will either remain unchanged or increase by 1.
\end{proof}

\subsection{Expected Area and Expected Size of Tukey Layers}
The following lemma shows the relation between the expected size and the expected area of the Tukey layers.

\begin{lemma}
\label{area and size}
Let $C\subseteq\R^2$ be a bounded and closed convex set of unit area and $X$ be the set of $n$ points chosen independently and uniformly from $C$. Then
\[
\E \left|U_{[t]}(X)\right|\leq n\left[1-\E A(S_{t+1}(X))\right].
\]
\end{lemma}

\begin{proof}
On the one hand, by Lemma~\ref{lem:shell position}, the points in $U_{[t]}(X)$ must be outside $S_{t+1}(X)$. On the other hand, there might be points of $X\setminus U_{[t]}(X)$ not lying in $S_{t+1}(X)$, either. Since those points not belonging to  $S_{t+1}(X)$ are uniform in $C\setminus S_{t+1}(X)$,  in expectation we have
\[\E \left|U_{[t]}(X)\right|\leq n\E \left[1-A(S_{t+1}(X))\right]=n\left[1-\E A(S_{t+1}(X))\right].\qedhere\]
\end{proof}

\subsection{Upper (Lower) Hull of Tukey Layer}\label{sec:upper_lower_hulls}
For a general convex polygon, let $P_1$ be the vertex with the smallest $x$-coordinate and $Q_1$ the vertex with the largest $x$-coordinate, where we break the tie by choosing the point with the largest $y$-coordinate for both points. Then, the \emph{upper hull} refers to the boundary of the polygon from $P_1$ to $Q_1$ in the clockwise orientation. Similarly, let $P_2$ be the vertex with the smallest $x$-coordinate and $Q_2$ the vertex of the largest $x$-coordinate of the polygon, where we break the tie by choosing the point with the smallest $y$-coordinate. It may happen that $P_1=P_2$ and $Q_1=Q_2$. The \emph{lower hull} refers to the boundary from $Q_2$ to $P_2$ in the clockwise orientation.

For a point $P$, if the ray ejecting vertically downwards (upwards) from $P$ crosses the upper (lower) hull of the convex polygon,  we shall say it is above (below) the upper (lower) hull.

% In Lemma~\ref{lem:convex position}, we proved that the points of $U_t(X)$ are in the position of a convex polygon, and therefore $U_t(X)$ exactly consists of the vertices of $S_t(X)$. The definition of upper and lower hulls can be applied to the Tukey layers. For a point $P$, if the ray ejecting vertically downwards (upwards) from $P$ crosses the upper (lower) hull of $S_t(X)$,  we shall say it is above (below) the upper (lower) hull.

% !TeX root = main.tex
\section{Upper Bound on Expected Size of Tukey Layers}
\label{section: upper bound for Tukey layer}

In this section, we prove $\E\left|U_{[t]}(X)\right| = O(kt\log(n/k))$, when the $n$ points of $X$ are sampled independently and uniformly from a convex $k$-gon. Our proof is inspired by \cite{har2011expected} in which Har-Peled considered the expected size of the convex hull of $X$ for $X$ uniformly sampled from a triangle of unit area. He partitions the triangle into $n\times n$ equal-area cells and gives a lower bound on the expected number of cells that are inside the convex hull. Dividing the lower bound by $n^2$ would yield a lower bound on the expected area of the convex hull, denoted by $\E A(X)$. Then by $\E\left|V(X)\right| = n\left[1-\E A(X)\right]$ from \cite{efron1965convex}, an upper bound on the expected size  $\E|V(X)|$ of the convex hull follows. The case where $X$ is uniformly sampled from a convex $k$-gon can be reduced to triangles by partitioning the $k$-gon into $k$ triangles. Before proving our main results, we need the following auxiliary lemma.

%By definition~\ref{def::t-shell}, a point $p$ is in $U_{[t]}(X)$ if and only if there exists a line $\ell$ through $p$, such that on one side of $\ell$ there are at most $(t-1)$ points. By this observation, we can divide the derive an upper bound on the number of points whose shell order are at most $(n-t)$. 

\begin{lemma}
\label{lem:distribution}
Given a point $p\in X$, the plane is partitioned into four open quadrants by the horizontal and vertical lines through $p$, as shown in Figure~\ref{fig:shell order}. 
If both the upper-left and upper-right quadrants contain at least $t$ points of $X$, then for any non-vertical line $\ell$ through $p$, there must be at least $t$ points of $X$ above $\ell$. In other words, the point $p$ cannot be above the upper hull of $S_t(X)$.
\end{lemma}

\begin{figure}[H]
\centering
\begin{minipage}[t]{0.9\textwidth}
\centering
\begin{tikzpicture}
\fill[gray!20](0,0) -- (2,0) -- (2,2) -- (0,2) -- (0,0);

\fill[black!20](0,0) -- (-2,0) -- (-2,2) -- (0,2) -- (0,0);

\draw node[fill,circle,minimum width=2pt, inner sep=0pt] at (0,0){};

\node at (0,0) [label=below left:$p$]{};

\draw (-2,0) -- (0,0) -- (2,0);

\draw (-1.8, 1.2) -- node[below]{$\ell$}(-0.3,0.2) -- (1.8,-1.2);

\draw (0,-2) -- (0,0) -- (0,2);
\end{tikzpicture}
\caption{The plane is divided into 4 open quadrants by the horizontal and vertical lines through $p$. The upper left and upper right quadrants are marked by dark grey and light grey colour, respectively. Line $\ell$ is an arbitrary non-vertical line through $p$. % The upper half-plane above $\ell$ always contains a complete sector and thus contains at least $t$ points.
}
\label{fig:shell order}
\end{minipage}
\end{figure}
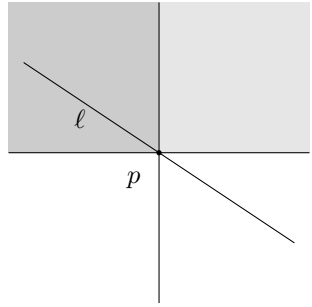

\begin{proof}
For any non-vertical line $\ell$ through $p$, either the upper-left or the upper-right quadrant is completely above $\ell$. Since both quadrants contain at least $t$ points, there are always $t$ points above $\ell$. By Definition~\ref{def::Tukey depth}, we know that $p$ cannot be above the upper hull of the $t$-th Tukey layer.
\end{proof}

Since the points in $X$ are chosen uniformly at random, we may assume that no three points are collinear and no two points have the same $x$ or $y$ coordinate, because such degenerate cases happen with zero probability. We decompose the convex hull into an upper hull and a lower hull, as defined in Section~\ref{sec:upper_lower_hulls}. Lemma~\ref{lem:distribution} implies that 
\begin{multline*}
\Pr\left(\text{$p$ is below the upper hull of $U_t(X)$}\right) \\
\geq \Pr\left(\text{$p$ has at least $t$ points in both upper-left and upper-right quadrants}\right)
\end{multline*}
and similarly
\begin{multline*}
\Pr\left(\text{$p$ is above the lower hull of $U_t(X)$}\right) \\
\geq \Pr\left(\text{$p$ has at least $t$ points in both lower-left and lower-right quadrants}\right).
\end{multline*}
Then we can upper bound $\Pr\left(p\in X\setminus U_{[t]}(X)\right)$ as
\begin{align*}
&\quad \Pr\left(p\in U_{[t]}(X)\right) \\
&= \Pr\left(\text{$p$ is on or above the upper hull of $U_t(X)$}\right) \\
&\qquad\ + \Pr\left(\text{$p$ is on or below the lower hull of $U_t(X)$}\right)\\
&= (1 - \Pr\left(\text{$p$ is below the upper hull of $U_t(X)$}\right)) \\
&\qquad\ + (1 - \Pr\left(\text{$p$ is above the lower hull of $U_t(X)$}\right)),
\end{align*}
whence an upper bound on $\Pr\left(p\in U_{[t]}(X)\right)$ would follow. Multiplying the upper bound by $n$ would finally produce an upper bound on $\E\left|U_{[t]}(X)\right|$.

\begin{theorem}
\label{thm:triangle upper bound}
Let $X$ be a set of $n$ points sampled independently and uniformly from a triangle, then $\E\left|U_{[t-1]}(X)\right|\leq 4t\ln{n} + 4t + 10$.
\end{theorem}

\begin{figure}[t]
	\centering
	\begin{minipage}[t]{\textwidth}
		\centering
		\begin{tikzpicture}[scale=0.8,every node/.style={scale=0.7}]
		\draw (0,0) -- (1.5,4.5) -- (6,0) -- (0,0);
		
		\draw[dashed] (1.2,0) -- (1.5,4.5);
		
		\draw[dashed] (2.4,0) -- (1.5,4.5);
		
		\draw[dashed] (3.6,0) -- (1.5,4.5);
		
		\draw[dashed] (4.8,0) -- (1.5,4.5);
		
		\draw[dashed] (0.829179608,2.487538821) -- (3.512461179,2.487538821);
		
		\draw[dashed] (0.551316702,1.653950106) -- (4.346049894,1.653950106);
		
		\draw[dashed] (0.338104997,1.014314988) -- (4.985685012,1.014314988);
		
		\draw[dashed] (0.158359214,0.475077641) -- (5.52492236,0.475077641);
		
		\node at (0.6,0) [label=below:col 1]{};
		
		\node at (1.8,0) [label=below:col 2]{};
		
		\node at (3,0) [label=below:col 3]{};
		
		\node at (4.2,0) [label=below:col 4]{};
		
		\node at (5.4,0) [label=below:col 5]{};
		
		\node at (0.079179608,0.237538821) [label=left:row 5]{};
		
		\node at (0.248232105,0.744696315) [label=left:row 4]{};
		
		\node at (0.444710849,1.334132547) [label=left:row 3]{};
		
		\node at (0.690248154,2.070744464) [label=left:row 2]{};
		
		\node at (1.164589803,3.493769411) [label=left:row 1]{};
		
		\node at (0.7,-0.1) [label=above:9]{};
		
		\node at (1.8,-0.1) [label=above:10]{};
		
		\node at (4.1,-0.1) [label=above:19]{};
		
		\node at (5.2,-0.1) [label=above:20]{};
		
		\node at (0.8,0.4) [label=above:7]{};
		
		\node at (1.75,0.4) [label=above:8]{};
		
		\node at (3.8,0.4) [label=above:17]{};
		
		\node at (4.7,0.4) [label=above:18]{};
		
		\node at (0.9,0.95) [label=above:5]{};
		
		\node at (1.7,0.95) [label=above:6]{};
		
		\node at (3.4,0.95) [label=above:15]{};
		
		\node at (4.3,0.95) [label=above:16]{};
		
		\node at (1,1.7) [label=above:3]{};
		
		\node at (1.65,1.7) [label=above:4]{};
		
		\node at (2.95,1.7) [label=above:13]{};
		
		\node at (3.6,1.7) [label=above:14]{};
		
		\node at (1.2,2.6) [label=above:1]{};
		
		\node at (1.6,2.6) [label=above:2]{};
		
		\node at (2.4,2.6) [label=above:11]{};
		
		\node at (2.83,2.6) [label=above:12]{};
		
		\node at (2.95,-0.1) [label=above:25]{};
		
		\node at (2.75,0.4) [label=above:24]{};
		
		\node at (2.55,0.95) [label=above:23]{};
		
		\node at (2.3,1.7) [label=above:22]{};
		
		\node at (2,2.6) [label=above:21]{};
		\end{tikzpicture}
		\caption{Partitioning of a triangle into $n^2$ equal-area cells for $n=5$. The cells are numbered for $j=3$ by Eq.~\eqref{eqn:numbering}.}
		\label{fig:cells}
	\end{minipage}
\end{figure}
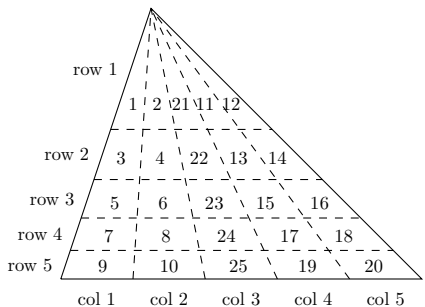

Denote the triangle by $T$ and, without loss of generality, assume that $T$ has area $1$. We partition $T$ into $n$ equal-area triangles by segments emanating from a fixed vertex. Each triangle is further partitioned into one triangle and  $(n-1)$ trapezoids with equal-area by line segments parallel to the opposite side. See Figure~\ref{fig:cells} for an illustration. There are thus  $n^2$ cells in $T$, each has area $1/n^2$. Let $G_{i,j}$ denote the cell in the $i$-th row and $j$-th column. We also define $G_{[i_1,i_2],[j_1,j_2]}=\bigcup_{i'=i_1}^{i_2}\bigcup_{j'=j_1}^{j_2} G_{i',j'}$. % When $i_1=i_2$ or  $j_1=j_2$, we abbreviate it as $G_{i,[j_1,j_2]}$ and $G_{[i_1,i_2],j}$, respectively.

\begin{proof}[Proof of Theorem~\ref{thm:triangle upper bound}]
We shall count the expected number of cells in each column that are above (resp.\ below) or intersecting the upper (resp.\ lower) hull of $S_t(X)$, the $t$-th Tukey layer. Summing up all those values will lead to an upper bound on the expected number of cells, and thus the expected area, outside $S_t(X)$. Since in an $n\times n$ grid, the boundary of a convex polygon can intersect at most $4n$ cells in total, we only need to count how many cells in the $j$-th column are above the upper hull of $S_t(X)$.

To count the expected number of cells above the upper hull of $S_t(X)$, let $Z_j$ ($1<j<n$) denote the maximum $i$ such that  $G_{ij}$ is above the upper hull of $S_t(X)$ and we shall find an upper bound on $\E[Z_j]$. Let $I_1$ (resp.\ $I_2$) be the row index of the $t$-th point from top to bottom in $G_{[1, n],[1,j-1]}$ (resp.\ $G_{[1, n],[j+1,n]}$). Then for any $G_{i,j}$ with $i>\max(I_1, I_2)$, there must be at least $t$ points in its upper left quadrant and also $t$ points in its upper right quadrant. By Lemma~\ref{lem:distribution}, such a point cannot be above the upper hull of $S_t(X)$. Therefore, $Z_j\leq\max\left(I_{1}, I_{2}\right)\leq I_{1} + I_{2}$ and thus $\E Z_j\leq \E I_{1} + \E I_{2}$. We can prove that $\E I_{1}\leq\frac{tn}{j-1}+1$ and $\E I_{2}\leq\frac{tn}{n-j}+1$ (the proof is postponed to Lemma~\ref{lmm:left_right_upper}), then
\[
\E Z_{j}\leq\E I_{1}+\E I_{2}\leq \frac{tn}{j-1}+\frac{tn}{n-j}+2.
\]
%Therefore at the top of the $j$-th column, the expected number of cells above the upper hull of $S_{t}(X)$ is at most $\frac{tn}{j-1}+\frac{tn}{n-j}+2$. 

To count the expected number of cells below the lower hull of $S_t(X)$, we analogously define $Z_j'$ to be the maximum $i$ such that $G_{n-i+1,j}$ is below or intersects the lower hull of $S_t(X)$. A similar argument to the above shows the same upper bound on $\E[Z_j']$, that is, 
\[
\E Z_{j}'\leq \frac{tn}{j-1}+\frac{tn}{n-j}+2.
\]
Note that the first and the last column each contains at most $n$ cells outside $S_t(X)$. The expected number of cells in $T$ which are outside $S_{t}(X)$ is therefore at most
\begin{align*}
2n + \sum_{j=2}^{n-1} \left(\E Z_j + \E Z_j'\right) &\leq 2n+2\cdot\sum_{j=2}^{n-1}\left(\frac{tn}{j-1}+\frac{tn}{n-j}+2\right) \\
&\leq 2n + 2\left[2tn\ln(n-2)+ 2tn + 2(n-2) \right] \\
&\leq 4tn\ln{n} + 4tn + 6n,
\end{align*}
together with the at most $4n$ cells that intersect the boundary of the $t$-th Tukey layer $S_t(X)$, when $n\geq 4$. It follows that
\[
\E A(S_{t}(X))\geq 1-\frac{4tn\ln{n}+4tn+ 6n+4n}{n^2}\geq 1-\frac{4t\ln{n} + 4t +10}{n}.
\]
By Lemma~\ref{area and size}, we finally conclude that $\E\left|U_{[t-1]}\right|\leq 4t\ln{n} + 4t + 10$ when $n\geq 4$. When $n < 4$, this bound holds trivially since $\E\left|U_{[t-1]}\right|\leq n$.
\end{proof}

\begin{lemma}
\label{lmm:left_right_upper}
Suppose that $1<j<n$. Let  $I_1$ (resp.\ $I_2$) be the row indices of the $t$-th point from top to bottom in $G_{[1, I_1],[1,j-1]}$ (resp.\ $G_{[1, I_2],[j+1,n]}$), then $\E I_{1}\leq \frac{tn}{j-1}+1$ and $\E I_{2}\leq \frac{tn}{n-j}+1$.
\end{lemma}

\begin{proof}
We prove $\E I_{1}\leq \frac{tn}{j-1}+1$ below, and a similar argument will give $\E I_{2}\leq \frac{tn}{n-j}+1$. We number the $n^2$ cells from $1$ to $n^2$ as follows. For a cell $G_{i,\ell}$, we define its number 
\begin{equation}\label{eqn:numbering}
	idx(G_{i,\ell}) = \begin{cases}
				(j-1)(i-1) + \ell, & \ell < j;\\
				(j-1)n + (n-j)(i-1) + \ell, & \ell > j;\\
				(n-1)n+i, & \ell = j.
			   \end{cases}	
\end{equation}
See Figure~\ref{fig:cells} for an illustration. Intuitively, the triangle is split into three parts, left to the $j$-th column, right to the $j$-the column and the $j$-th column. In each part the cells are numbered one by one from left to right and from top to bottom; overall, the left part precedes the right part and the right part precedes the $j$-th column. 

Now, we can refer to each cell by its number and denote the cells by $G_1,\dots,G_{n^2}$, abusing the notation. Since all cells have the same area, a uniform random point in the triangle $T$ can be generated by first choosing an integer in $m\in \{1,\dots,n^2\}$ uniformly at random and then generating a uniform random point in $G_m$. Also we denote by $|G_m|$ the number of points in $X$ that are contained in $G_m$.

Let $h$ be the integer such that $\sum_{i=1}^{h-1} |G_i| < t$ and $\sum_{i=1}^{h} |G_i| \geq t$. This is exactly the $t$-th smallest integer among $n$ uniform samples from $\{1,\dots,n^2\}$. Let $f_t(x)$ be the density function of the $t$-th smallest value among $n$ independent uniform points in $[0,1]$. Then
\begin{align*}
\E h = \int_0^1 \lceil x n^2\rceil f_t(x) dx \leq \int_0^1 (xn^2+1)f_t(x)dx &= n^2\int_0^1 xf_t(x)dx + 1 \\
&= n^2\frac{t}{n+1} + 1 \\
&\leq t(n-1) + 1\\
&\leq tn.
\end{align*}
Here we used the fact that $\int_0^1 x f_t(x) dx = \frac{t}{n+1}$. The integral is the expected value of the $t$-th smallest value among $n$ independent uniform points in $[0,1]$, and it is a classic result that this expected value is exactly $t/(n+1)$ (see, e.g., \cite[Lemma 8.3]{mitzenmacher2017probability}).

When $h\leq n(j-1)$, we have $I_1=\lceil h / (j-1)\rceil$. When $h > n(j-1)$, it automatically holds that $I_1 \leq n \leq h/(j-1)$. In both cases, we have $I_1 \leq \lceil h / (j-1)\rceil$. Therefore,
\[
\E I_1 \leq \E \left\lceil \frac{h}{j-1}\right\rceil \leq \frac{\E h}{j-1}+1 \leq \frac{tn}{j-1}+1.\qedhere
\]
\end{proof}

\begin{theorem}
\label{thm:upper bound}
Let $X$ be a set of $n$ points sampled independently and uniformly from a convex $k$-gon. Then we have $\E\left|U_{[t-1]}(X)\right|\leq 4tk\ln(n/k)+4tk+10k$.
\end{theorem}

\begin{proof}
Partition the convex $k$-gon into $k$ triangles. Let $X_{1}, X_2,\dots,X_k$ be the set of points of $X$ in the triangles and $n_i = |X_i|$ for $i=1,\dots, k$. Note that $n_{1}, n_2, \dots, n_k$ are random numbers subject to $\sum_{i=1}^{k}n_{i}=n$. It follows from Corollary~\ref{coro:shell t-part} that
\begin{align*}
\E\left[U_{[t-1]}(X)|n_{1}, n_{2}, ..., n_{k}\right] \leq \sum_{i=1}^{k}\E\left[U_{[t-1]}(X_{i})|n_{i}\right] & \leq \sum_{i=1}^{k}(4t\ln{n_i} + 4t+10)\\
& = 4t\sum_{i=1}^{k}\ln{n_i} + 4tk + 10k\\
& \leq 4tk\ln(n/k)+4tk+10k.\ \qedhere\\
\end{align*}
%The claimed result follows immediately. 
%At last, we conclude that $\E U_{[t]}(X)=O(kt\log{n})$.
\end{proof}

% !TeX root = main.tex
\begin{figure}[t]
	\centering
	\begin{minipage}[h]{0.45\textwidth}
		\centering
		\begin{tikzpicture}
		\draw (0,0) -- (4,0) -- (0,4) -- (0,0);
		
		\draw (4/3,4/3) -- (0,2);
		
		\draw (4/3,4/3) -- (2,0);
		
		\draw (4/3,4/3) -- (2,2);
		
		\node at (0.5,0.5) {$R_1$};
		
		\node at (2.5,0.5) {$R_2$};
		
		\node at (0.5,2.5) {$R_3$};
		
		\node at (0,0) [label=below left:{$(0,0)$}]{};
		
		\node at (4,0) [label=below:{$(1,0)$}]{};
		
		\node at (2,0) [label=below:{$\left(\frac{1}{2},0\right)$}]{};
		
		\node at (0,2) [label=left:{$\left(0,\frac{1}{2}\right)$}]{};
		
		\node at (2,2) [label=right:{$\left(\frac{1}{2},\frac{1}{2}\right)$}]{};
		
		\node at (0,4) [label=left:{$(0,1)$}]{};
		
		\node at (4/3,4/3) [label=above:{$\left(\frac{1}{3},\frac{1}{3}\right)$}]{};
		\end{tikzpicture}
		\caption{The triangle is divided into three parts, by connecting the centroid to the midpoint of each edge.}
		\label{fig:partition}
	\end{minipage}
	\hfill
	\begin{minipage}[h]{0.45\textwidth}
		\centering
		\begin{tikzpicture}
		\draw node[fill,circle,minimum width=2pt, inner sep=0pt] at (1,1){};
		
		\node at (1,1) [label=above right:$p$]{};
		
		\draw (0,0) -- (4,0) -- (0,4) -- (0,0);
		
		\draw[dotted,thick] (-0.25,1) -- (3.25,1);
		
		\draw[dotted,thick] (1,-0.25) -- (1,3.25);
		
		\node at (0,0) [label=below left:{$(0,0)$}]{};
		
		\node at (4,0) [label=below:{$(1,0)$}]{};
		
		\node at (0,4) [label=left:{$(0,1)$}]{};
		\end{tikzpicture}
		\caption{By the horizontal line and the vertical line through a given point $p$, the triangle is divided into four quadrants.}
		\label{fig:quadrants}
	\end{minipage}
\end{figure}

\section{Upper Bound on Expected Size of Convex Layers}
\label{section: upper bound for convex layer}
In this section, we shall prove an upper bound $O\left(kt^{3}\log{\frac{n}{kt^2}}\right)$ on $\mathbb{E}\left|V_{[t]}(X)\right|$, when $X$ is sampled uniformly from a convex $k$-gon. The proof is inspired by \cite{dwyer1988convex} and \cite{he2018maximal}. We first consider the case where the points in $X$ are sampled uniformly from a triangle $T$ and obtain an upper bound $O\left(kt^{3}\log{\frac{n}{kt^2}}\right)$, which, by Corollary~\ref{coro:layer t-part}, implies an upper bound $O\left(kt^{3}\log{\frac{n}{kt^2}}\right)$ when $X$ is sampled from a $k$-gon. The problem can be further reduced to finding an upper bound on the probability $\Pr\left(p\in V_{[t]}(X)\right)$ for a single point $p\in X$, which, multiplied by $n$, will be an upper bound on $\mathbb{E}\left|V_{[t]}(X)\right|$. 

\begin{theorem}
\label{thm:convex layer upper bound}
Let $X$ be a set of n points sampled independently and uniformly from a triangle $T$, then $\mathbb{E}\left|V_{[t]}(X)\right| = O\left(t^3 \log(n/t^2)\right)$.
\end{theorem}

\begin{proof}%[Proof of Theorem~\ref{thm:convex layer upper bound}]
As the combinatorial properties of convex hulls are affine invariant, we may assume the vertices of $T$ are $(0,0)$, $(1,0)$ and $(0,1)$. We partition $T$ into three regions $R_1, R_2, R_3$ with equal area by connecting the centroid $\left(\frac{1}{3}, \frac{1}{3}\right)$ to the midpoint of each edge (see Figure~\ref{fig:partition}). Then
$\Pr\left(p\in V_{[t]}(X)|p\in R_i\right)$ are all equal for $i=1,2,3$ and so
\begin{align*}
\Pr\left(p\in V_{[t]}(X)\right) &=\sum_{i=1}^{3}\Pr\left(p\in V_{[t]}(X)|p\in R_i\right) \Pr(p\in R_{i})\\
&=\sum_{i=1}^{3}\Pr\left(p\in V_{[t]}(X)|p\in R_i\right)\cdot\frac{1}{3}\\
&= \Pr\left(p\in V_{[t]}(X)|p\in R_1\right).
\end{align*}

\begin{figure}[t]
\centering
\begin{minipage}[h]{0.24\textwidth}
\centering
\begin{tikzpicture}
\fill[black!20](0,0) -- (1/4,0) -- (1/4,1/4) -- (0,1/4) -- (0,0);

\fill[black!20](1/4,1/4) -- (1/4,1/2) -- (1/2,1/2) -- (1/2,1/4) -- (1/4,1/4);

\fill[black!20](1/2,1/2) -- (3/4,1/2) -- (3/4,3/4) -- (1/2,3/4) -- (1/2,1/2);

\fill[black!20](3/4,3/4) -- (1,3/4) -- (1,1) -- (3/4,1) -- (3/4,3/4);

\draw (0,0) -- (1,0) -- (1,1) -- (0,1) -- (0,0);

\draw[dashed] (0,1/4) -- (1,1/4);

\draw[dashed] (0,2/4) -- (1,2/4);

\draw[dashed] (0,3/4) -- (1,3/4);

\draw[dashed] (1/4,0) -- (1/4,1);

\draw[dashed] (2/4,0) -- (2/4,1);

\draw[dashed] (3/4,0) -- (3/4,1);
\end{tikzpicture}
\end{minipage}
\begin{minipage}[h]{0.24\textwidth}
\centering
\begin{tikzpicture}
\fill[black!20](1,0.708712153) -- (1.572821962,0.708712153) -- (3/2,1) -- (1,1) -- (1,0.708712153);

\fill[black!20](1.572821962,0.708712153) -- (2.145643924,0.708712153) -- (2.274754878,0.450490243) -- (1.637377439,0.45049024) -- (1.572821962,0.708712153);

\fill[black!20](2.274754878,0.450490243) -- (2.912132318,0.450490243) -- (3.087911636,0.216117819) -- (2.391941091,0.216117819) -- (2.274754878,0.450490243);

\fill[black!20](3.087911636,0.216117819) -- (3.783882181,0.216117819) -- (4,0) -- (13/4,0) -- (3.087911636,0.216117819);

\draw (1,0) -- (1,1) -- (3,1) -- (4,0) -- (1,0);

\draw[dashed] (1,0.708712153) -- (3.291287847,0.708712153);

\draw[dashed] (1,0.450490243) -- (3.549509757,0.450490243);

\draw[dashed] (1,0.216117819) -- (3.783882181,0.216117819);

\draw[dashed] (1,0.216117819) -- (3.783882181,0.216117819);

\draw[dashed] (3/2,1) -- (7/4,0);

\draw[dashed] (2,1) -- (5/2,0);

\draw[dashed] (5/2,1) -- (13/4,0);
\end{tikzpicture}
\end{minipage}
\begin{minipage}[h]{0.24\textwidth}
\centering
\begin{tikzpicture}
\fill[black!20](1,1) -- (9/7,1) -- (9/7,9/7) -- (1,9/7) -- (1,1);

\fill[black!20](9/7,9/7) -- (11/7,9/7) -- (11/7,11/7) -- (9/7,11/7) -- (9/7,9/7);

\fill[black!20](11/7,11/7) -- (11/7,13/7) -- (13/7,13/7) -- (13/7,11/7) -- (11/7,11/7);

\fill[black!20](13/7,13/7) -- (13/7,15/7) -- (15/7,13/7) -- (13/7,13/7);

\draw (1,1) -- (3,1) -- (1,3) -- (1,1);

\draw[dashed] (9/7,1) -- (9/7,19/7);

\draw[dashed] (11/7,1) -- (11/7,17/7);

\draw[dashed] (13/7,1) -- (13/7,15/7);

\draw[dashed] (15/7,1) -- (15/7,13/7);

\draw[dashed] (17/7,1) -- (17/7,11/7);

\draw[dashed] (19/7,1) -- (19/7,9/7);

\draw[dashed] (1,9/7) -- (19/7,9/7);

\draw[dashed] (1,11/7) -- (17/7,11/7);

\draw[dashed] (1,13/7) -- (15/7,13/7);

\draw[dashed] (1,15/7) -- (13/7,15/7);

\draw[dashed] (1,17/7) -- (11/7,17/7);

\draw[dashed] (1,19/7) -- (9/7,19/7);
\end{tikzpicture}
\end{minipage}
\begin{minipage}[h]{0.24\textwidth}
\centering
\begin{tikzpicture}
\fill[black!20](0.708712153,1) -- (0.708712153,1.572821962) -- (1,3/2) -- (1,1) -- (0.708712153,1);

\fill[black!20](0.708712153,1.572821962) -- (0.708712153,2.145643924) -- (0.450490243,2.274754878) -- (0.45049024,1.637377439) -- (0.708712153,1.572821962);

\fill[black!20](0.450490243,2.274754878) -- (0.450490243,2.912132318) -- (0.216117819,3.087911636) -- (0.216117819,2.391941091) -- (0.450490243,2.274754878);

\fill[black!20](0.216117819,3.087911636) -- (0.216117819,3.783882181) -- (0,4) -- (0,13/4) -- (0.216117819,3.087911636);

\draw (0,1) -- (0,4) -- (1,3) -- (1,1) -- (0,1);

\draw[dashed] (0.708712153,1) -- (0.708712153,3.291287847);

\draw[dashed] (0.450490243,1) -- (0.450490243,3.549509757);

\draw[dashed] (0.216117819,1) -- (0.216117819,3.783882181);

\draw[dashed] (0.216117819,1) -- (0.216117819,3.783882181);

\draw[dashed] (1,3/2) -- (0,7/4);

\draw[dashed] (1,2) -- (0,5/2);

\draw[dashed] (1,5/2) -- (0,13/4);
\end{tikzpicture}
\end{minipage}
\caption{Partition of each quadrant of the triangle into cells when $t=4$. In each single quadrant, the cells have the equal area. There are exactly $t$ diagonal cells in each quadrant, marked in grey colour.}
\label{fig:quadrant cell}
\end{figure}
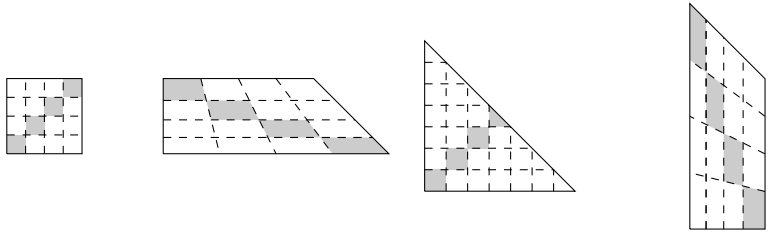

We turn to find an upper bound on $ \Pr\left(p\in V_{[t]}(X)|p\in R_1\right)$. For this purpose, the triangle $T$ is divided into four quadrants by a vertical and a horizontal line through $p$ as shown in Figure~\ref{fig:quadrants}. Each quadrant is further partitioned into multiple cells as in Figure~\ref{fig:quadrant cell}. The triangular quadrant is partitioned into $(2t+1)t$ cells by $(2t-1)$ equally spaced horizontal lines and another $(2t-1)$ equally spaced vertical lines. Each of the other three quadrants are partitioned into $t^2$ equal-area cells. This construction ensures exactly $t$ diagonal cells in each of the four quadrants. 

We claim that if $p \in V_{[t]}(X)$, then at least one of the $4t$ diagonal cells must be empty. The proof of this claim is deferred to Lemma~\ref{lem:diag}. By this observation, the probability of $p \in V_{[t]}(X)$ is at most the probability that at least one of the $4t$ diagonal cells is empty, which we upper bound as follows. Let $(p_1, p_2)$ denote the coordinates of $p$. When $p\in R_1$, the area of each quadrant is at least $\frac{1}{2}p_1p_2$ by \cite[Section 2]{dwyer1988convex} and the probability mass (with respect to the uniform distribution on $T$) of each quadrant is at least $p_1p_2$. Therefore each diagonal cell has probability mass at least $\frac{p_1p_2}{4t^2}$, and the expected number of points in every single cell is at least $\frac{np_1p_2}{4t^2}$. By the multiplicative form of Chernoff bound \cite[Theorem 4.5]{mitzenmacher2017probability}, the probability that a diagonal cell is empty is at most $\exp\left(-\frac{np_{1}p_{2}}{16t^2}\right)$. Further by a union bound, the probability that at least one of the $4t$ diagonal cells is empty in triangle $T$ is at most $4t\exp\left(-\frac{np_{1}p_{2}}{16t^2}\right)$. Therefore,
\[
\Pr\left(p\in V_{[t]}(X)|p_{1}p_{2}=y, p\in R_{1}\right)\leq 4te^{-\frac{ny}{16t^2}},
\]
whence we can show that
\[
\Pr\left(p\in V_{[t]}(X)|p\in R_1\right)\leq 12t \int_{0}^{1/9} e^{-\frac{ny}{16t^2}} \log{\frac{1}{y}}\dd y = 12t \cdot O\left(\frac{t^2}{n}\log{\frac{n}{t^2}}\right),
\]
whose proof is postponed to Lemma~\ref{eq:upper_1} and Lemma~\ref{eq:upper_2}. 
It follows that $\Pr\left(p \in V_{[t]}(X)\right)=O\left(\frac{t^3}{n}\log{\frac{n}{t^2}}\right)$ for any $p \in X$ and, finally, that $\mathbb{E}\left|V_{[t]}(X)\right| = O\left(t^3\log{\frac{n}{t^2}}\right)$.
\end{proof}

Now we are ready to prove the following main theorem.
\begin{theorem}
\label{thm:upper bound 2}
Let $X$ be a set of $n$ points sampled independently and uniformly from a convex $k$-gon, then we have $\mathbb{E}\left|V_{[t]}(X)\right|=O\left(kt^3\log\frac{n}{kt^2}\right)$.
\end{theorem}

\begin{proof}
As in the proof of Theorem~\ref{thm:upper bound}, we partition the $k$-gon into $k$ triangles. Let $n_1, n_2, \dots, n_k$ denote the number of points in each triangle. It follows from Corollary~\ref{coro:layer t-part} that
\begin{align*}
\E\left[V_{[t]}(X)|n_{1}, n_{2}, \dots, n_{k}\right]\leq\sum_{i=1}^{k}\E\left[V_{[t]}(X_{i})|n_{i}\right]&\leq\sum_{i=1}^{k}O\left(t^3\log{\frac{n_i}{t^2}}\right)\\
&=O\left(kt^3\log{\frac{n}{kt^2}}\right),
\end{align*}
where we used the AM-GM inequality and the fact that $\sum_{i=1}^k n_i = n$ in the last step.
\end{proof}

In the rest of this section, we state and prove those lemmata used in the proof of Theorem~\ref{thm:convex layer upper bound}. We denote the density and the cumulative distribution functions of the product $p_1\cdot p_2$ by \(\rho_{p_{1}p_{2}}(\cdot)\) and \(F_{p_{1}p_{2}}(\cdot)\), respectively.

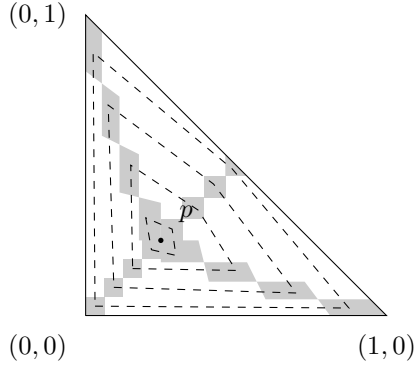
\begin{figure}[!tb]
\centering
\begin{minipage}[h]{0.9\textwidth}
\centering
\begin{tikzpicture}
\fill[black!20] (0,0) -- (1/4,0) -- (1/4,1/4) -- (0,1/4) -- (0,0);

\fill[black!20](1/4,1/4) -- (1/4,1/2) -- (1/2,1/2) -- (1/2,1/4) -- (1/4,1/4);

\fill[black!20](1/2,1/2) -- (3/4,1/2) -- (3/4,3/4) -- (1/2,3/4) -- (1/2,1/2);

\fill[black!20](0,0) -- (1/4,0) -- (1/4,1/4) -- (0,1/4) -- (0,0);

\fill[black!20](3/4,3/4) -- (1,3/4) -- (1,1) -- (3/4,1) -- (3/4,3/4);

\fill[black!20](1,0.708712153) -- (1.572821962,0.708712153) -- (3/2,1) -- (1,1) -- (1,0.708712153);

\fill[black!20](1.572821962,0.708712153) -- (2.145643924,0.708712153) -- (2.274754878,0.450490243) -- (1.637377439,0.45049024) -- (1.572821962,0.708712153);

\fill[black!20](2.274754878,0.450490243) -- (2.912132318,0.450490243) -- (3.087911636,0.216117819) -- (2.391941091,0.216117819) -- (2.274754878,0.450490243);

\fill[black!20](3.087911636,0.216117819) -- (3.783882181,0.216117819) -- (4,0) -- (13/4,0) -- (3.087911636,0.216117819);

\fill[black!20](1,1) -- (9/7,1) -- (9/7,9/7) -- (1,9/7) -- (1,1);

\fill[black!20](9/7,9/7) -- (11/7,9/7) -- (11/7,11/7) -- (9/7,11/7) -- (9/7,9/7);

\fill[black!20](11/7,11/7) -- (11/7,13/7) -- (13/7,13/7) -- (13/7,11/7) -- (11/7,11/7);

\fill[black!20](13/7,13/7) -- (13/7,15/7) -- (15/7,13/7) -- (13/7,13/7);

\fill[black!20](0.708712153,1) -- (0.708712153,1.572821962) -- (1,3/2) -- (1,1) -- (0.708712153,1);

\fill[black!20](0.708712153,1.572821962) -- (0.708712153,2.145643924) -- (0.450490243,2.274754878) -- (0.45049024,1.637377439) -- (0.708712153,1.572821962);

\fill[black!20](0.450490243,2.274754878) -- (0.450490243,2.912132318) -- (0.216117819,3.087911636) -- (0.216117819,2.391941091) -- (0.450490243,2.274754878);

\fill[black!20](0.216117819,3.087911636) -- (0.216117819,3.783882181) -- (0,4) -- (0,13/4) -- (0.216117819,3.087911636);

\draw node[fill,circle,minimum width=2pt, inner sep=0pt] at (1,1){};

\node at (1,1) [label=above right:$p$]{};

\draw (0,0) -- (4,0) -- (0,4) -- (0,0);

\node at (0,0) [label=below left:{$(0,0)$}]{};

\node at (4,0) [label=below:{$(1,0)$}]{};

\node at (0,4) [label=left:{$(0,1)$}]{};

\draw[dashed] (1/8,1/8) -- (3.5,0.1) -- (1.95,1.95) -- (0.1,3.5) -- (1/8,1/8);

\draw[dashed] (3/8,3/8) -- (2.8,0.3) -- (1.7,1.75) -- (0.3,2.8) -- (3/8,3/8);

\draw[dashed] (5/8,5/8) -- (2,0.6) -- (1.5,1.4) -- (0.6,2) -- (5/8,5/8);

\draw[dashed] (7/8,7/8) -- (1.2,0.8) -- (1.15, 1.15) -- (0.8,1.3) -- (7/8,7/8);
\end{tikzpicture}
\caption{The diagonal cells are shaded. Connecting one point in the diagonal cell of the same order in each quadrant forms a convex layer, marked by a dashed polyline.}
\label{fig:convex layers}
\end{minipage}
\end{figure}

\begin{lemma}
\label{lem:diag}
If $p \in V_{[t]}(X)$, there must be at least one empty diagonal cell. 
\end{lemma}

\begin{proof}
If none of the $4t$ diagonal cells is empty, we can construct $t$ convex layers enclosing $p$, where each layer consists of four points from the diagonal cells, one from each quadrant (see Figure~\ref{fig:convex layers}). The convex depth of $p$ is thus at least $(t+1)$. Although there may be more than one point in each diagonal cell, we know from Lemma~\ref{lem:dahal} that the convex depth of $p$ cannot decrease after those additional points are included. This contradicts the assumption that $p \in V_{[t]}(X)$. %and our assumption that none of the $4t$ diagonal cells is empty. 
Therefore, some diagonal cell must be empty.
\end{proof}

\begin{lemma}[{\cite[Theorem 1]{dwyer1988convex}}]
\label{lem:prob}
$F_{p_{1}p_{2}}\left(y |p \in R_1\right) \leq  3 F_{p_{1}p_{2}}\left(y| p \in [0,1] \times [0,1]\right)$.
\end{lemma}

\begin{lemma}[{\cite[section I.8]{Feller71}}]
\label{lem:feller}
$ \rho_{p_{1}p_{2}}\left(y| p \in [0,1] \times [0,1]\right)=\log(1/y).$
\end{lemma}

\begin{lemma}
\label{eq:upper_1}
If $\Pr\left(p\in V_{[t]}(X)|p_{1}p_{2}=y, p\in R_{1}\right)\leq 4te^{-\frac{ny}{16t^2}}$, then
\[\Pr\left(p\in V_{[t]}(X)|p\in R_1\right) \leq 12t\int_{0}^{1/9} e^{-\frac{ny}{16t^2 }} \log{\frac{1}{y}}\dd y.\]
\end{lemma}

\begin{proof}
It is easy to prove that $p_{1}p_{2}$ reaches its maximum value $\frac{1}{9}$ at $\left(\frac{1}{3}, \frac{1}{3}\right)$ for $p\in H_1$. Then we have
\begin{align*}
\Pr\left(p\in V_{[t]}(X)|p\in R_{1}\right) &= \int_{0}^{1/9}\Pr\left(p\in V_{[t]}(X)|p_{1}p_{2}=y, p\in R_{1}\right)\cdot \rho_{p_{1}p_{2}}\left(y|p \in R_1\right)\dd y \\
&\leq 4t\int_{0}^{1/9} e^{-\frac{ny}{16t^2}}\cdot \rho_{ p_{1}p_{2}} \left(y |p\in R_1\right)\dd y   \quad \\
&= 4t\int_{0}^{1/9} e^{-\frac{ny}{16t^2}}\dd  F_{p_{1}p_{2}} \left(y|p\in R_1\right).
\end{align*}
By Lemma~\ref{lem:prob} and Lemma~\ref{lem:feller},
\begin{align*}
\int_{0}^{1/9} e^{-\frac{ny}{16t^2}}\dd  F_{p_{1}p_{2}} \left(y|p\in R_1\right) &\leq 3\int_{0}^{1/9} e^{-\frac{ny}{16t^2}}\dd  F_{p_{1}p_{2}} \left(y|p\in [0,1] \times [0,1] \right) \\
&= 3\int_{0}^{1/9} e^{-\frac{ny}{16t^2}}\log{\frac{1}{y}} \dd y.
\end{align*}
thus 
\[\Pr\left(p\in V_{[t]}(X)|p\in R_{1}\right) \leq 12t\int_{0}^{1/9} e^{-\frac{ny}{16t^2}}\log{\frac{1}{y}} \dd y.\qedhere\]
\end{proof}

\begin{lemma}
\label{eq:upper_2}
$\int_{0}^{1/9} e^{-\frac{ny}{16t^2}}\log{\frac{1}{y}}\dd y = O \left( \frac{t^2}{n} \log \frac{n}{t^2} \right)$.
\end{lemma}

\begin{proof}
Substituting $y$ with $z = \frac{ny}{16t^2}$, we have
\begin{align*}
I & =\int_{0}^{1/9} e^{-\frac{ny}{16t^2}}\log{\frac{1}{y}}\dd y\\
& = \frac{16t^2}{n} \int_{0}^{1/9}  e^{-z} \left(\log \frac{n}{16t^2} + \log\frac{1}{z} \right)\dd  z \\
& \leq  \frac{16t^2}{n}  \log \frac{n}{16t^2}  \int_{0}^{\infty}  e^{ -z } \dd z +\frac{16t^2}{n}  \int_{0}^{\infty} e^{-z} \log \frac{1}{z} \dd  z. 
\end{align*}
Since both $\int_{0}^{\infty} e^{-z}  dz $ and $  \int_{0}^{\infty} e^{-z} \log \frac{1}{z} \dd  z$ are constants, we conclude
\[\int_{0}^{1/9} e^{-\frac{ny}{16t^2}}\log\frac{1}{y} \dd y = O \left( \frac{t^2}{n} \log \frac{n}{t^2} \right).\qedhere\]

\end{proof}

% !TeX root = main.tex
\section{Lower Bound of Expected Size of Tukey Layers}
\label{section: low bound for convex shell} 
We shall prove the lower bound on the expected size of $U_{[t]}(X)$, the first $t$ Tukey layers, for two special cases where $X$ is sampled from a parallelogram (Section~\ref{subsec:parallelogram lower bound}) and a triangle (Section~\ref{subsec:triangle lower bound}). We need the following lemma throughout this section.

\begin{lemma}[{\cite[Section 3]{affentranger1991convex}}]
	\label{lem:integration}
	For all integer \(r,s \geq 0\) and for all \(c \in (0,1]\) we have
	\begin{align*}
	\int_{0}^{1} \int_{0}^{1}(1-cxy)^{n-s}(xy)^r\mathrm{d}x\mathrm{d}y = \frac{r!}{c^{r+1}} \cdot \frac{\log n}{n^{r+1}}+ O\left(\frac{1}{n^{r+1}}\right),\quad n\to\infty.
	\end{align*}
\end{lemma}

\subsection{Parallelogram}
\label{subsec:parallelogram lower bound}
Without loss of generality, we may assume that the parallelogram is a unit square $[0,1]\times [0,1]$, because the combinatorial properties would not change under an affine transformation. For each point $p=(p_1, p_2)\in X$, we now compute the probability that it is on the first $t$ Tukey layers of $X$. For this purpose, we introduce the following definition.

\begin{definition}
Given a point $p=(p_1, p_2)$ with $0\leq p_1<\frac{1}{2}$ and $0\leq p_2<\frac{1}{2}$, the dividing line is defined to be
\[\ell_{0}:\frac{x}{2p_1}+\frac{y}{2p_2}=1.\]
The dividing line when $p_1\geq\frac{1}{2}$ or $p_2\geq\frac{1}{2}$ can be defined symmetrically.
\end{definition}
The line divides the unit square into  a triangle of area $2p_1p_2$ and a pentagon of area $(1-2p_1p_2)$. Notice that a sufficient condition for a point $p$ to be on the first $t$ Tukey layers is that, there are no more than $(t-1)$ points in the triangular part. We thus have the following theorem.

\begin{theorem}
\label{them:rectangle lower bound}
Suppose that $X$ consists of $n$ independent and uniformly sampled points from a unit square. There exists an absolute constant $\alpha > 0$ such that whenever $t\leq\alpha\sqrt{n}$, it holds that $\E\left|U_{[t]}(X)\right| = \Omega(t\log n)$ as $n\to\infty$. Furthermore, when $t = o((n/\log n)^{1/3})$, it holds that $\E\left|U_{[t]}(X)\right|\geq 2t\log n + O(1)$ as $n\to\infty$.
\end{theorem}
\begin{proof}
\begin{align*}
\Pr(p\in U_{[t]})&\geq  \Pr(\mbox{no more than } t \mbox{ points under the dividing line } \ell_{0})
\\
&=  4\int_{0}^{\frac{1}{2}}\int_{0}^{\frac{1}{2}}\sum_{i=0}^{t-1}\binom{n-1}{i}(2p_1p_2)^{i}(1-2p_1p_2)^{n-1-i}\dd p_1\dd p_2\\
&=  4\sum_{i=0}^{t-1}\binom{n-1}{i}\int_{0}^{\frac{1}{2}}\int_{0}^{\frac{1}{2}}(2p_1p_2)^{i}(1-2p_1p_2)^{n-1-i}\dd p_1\dd p_2\\
&= \sum_{i=0}^{t-1}\binom{n-1}{i}\int_{0}^{\frac{1}{2}}\int_{0}^{\frac{1}{2}}(2p_1p_2)^{i}(1-2p_1p_2)^{n-1-i}\dd (2p_1)\dd (2p_2)\\
&=  \sum_{i=0}^{t-1}\binom{n-1}{i}\int_{0}^{1}\int_{0}^{1}\left(\frac{p_1p_2}{2}\right)^{i}\left(1-\frac{p_1p_2}{2}\right)^{n-1-i}\dd p_1\dd p_2\\
&=  \sum_{i=0}^{t-1}\frac{1}{2^{i}}\binom{n-1}{i}\int_{0}^{1}\int_{0}^{1}(p_1p_2)^{i}\left(1-\frac{1}{2}p_1p_2\right)^{n-1-i}\dd p_1\dd p_2.
\end{align*}
By Lemma~\ref{lem:integration}, when $n\to\infty$, we have
\begin{align*}
\int_{0}^{1}\int_{0}^{1}(p_1p_2)^{i}\left(1-\frac{1}{2}p_1p_2\right)^{n-1-i}\dd p_1\dd p_2=\frac{i!}{\left(\frac{1}{2}\right)^{i+1}}\frac{\log{n}}{n^{i+1}}+O\left(\frac{1}{n^{i+1}}\right).
\end{align*}
Therefore, as $n\to\infty$,
\begin{align*}
\Pr(p\in U_{[t]}) &\geq \sum_{i=0}^{t-1}\frac{1}{2^{i}}\binom{n-1}{i}\left[\frac{i!}{\left(\frac{1}{2}\right)^{i+1}}\frac{\log{n}}{n^{i+1}}+O\left(\frac{1}{n^{i+1}}\right)\right]\\
&= \sum_{i=0}^{t-1}\left[2\cdot\frac{(n-1)!}{(n-1-i)!\cdot n^{i}}\cdot\frac{\log{n}}{n}+O\left(\frac{1}{2^{i}i!n}\right)\right]\\
&= \sum_{i=0}^{t-1}\left[\frac{2\log{n}}{n}\cdot\left(1-\frac{1}{n}\right)\cdot\left(1-\frac{2}{n}\right)\dots \cdot\left(1-\frac{i}{n}\right)+O\left(\frac{1}{2^{i}i!n}\right)\right]\\
&\geq \sum_{i=0}^{t-1}\frac{2\log{n}}{n}\cdot\left(1-\frac{(i+1)i}{2n}\right)+O\left(\frac{1}{n}\right)\\
&\geq \sum_{i=0}^{t-1}\frac{2\log{n}}{n}\cdot\left(1-\frac{(t-1)t}{2n}\right)+O\left(\frac{1}{n}\right)\\
&\geq\frac{2t\log{n}}{n}\left(1-\frac{t^2}{2n}\right)+O\left(\frac{1}{n}\right).
\end{align*}
Finally, the expected number of points on the first $t$ Tukey layers
\[
\mathbb{E}\left|U_{[t]}\right|=\sum_{p\in X}\Pr(p\in U_{[t]})\geq 2\left(1-\frac{t^2}{2n}\right)t\log{n} + O(1).
\]
The conclusions follow immediately.
\end{proof}

\subsection{Triangle}
\label{subsec:triangle lower bound}

\begin{theorem}
\label{them:triangle lower bound}
Suppose that $X$ consists of $n$ independent and uniformly sampled points from a triangle. There exists an absolute constant $\alpha > 0$ such that whenever $t\leq\alpha\sqrt{n}$, it holds that $\E\left|U_{[t]}(X)\right| = \Omega(t\log n)$ as $n\to\infty$.
\end{theorem}

\begin{proof}
The proof is similar to that of Theorem~\ref{them:rectangle lower bound}. Without loss of generality, we assume that the vertices of the triangle are $(0, 0)$, $(0, 1)$ and $(1, 0)$. Here we only consider those $p$ where $0\leq p_1\leq \frac{1}{2}$ and $0\leq p_2\leq\frac{1}{2}$. We now find a lower bound on $\Pr\left(p\in U_{[t]}, 0\leq p_1\leq\frac{1}{2}, 0\leq p_2\leq\frac{1}{2}\right)$. Note that dividing line divide the triangle into a triangle of area $2p_1p_2$ and a 
quadrilateral of area $\frac{1}{2}-2p_1p_2$. Their probability masses are $4p_1p_2$ and $1-4p_1p_2$ respectively.

\begin{align*}
\Pr(p\in U_{[t]}) &\geq \Pr\left(p\in U_{[t]}, 0\leq p_1\leq\frac{1}{2}, 0\leq p_2\leq\frac{1}{2}\right)\\
& \geq \Pr(\mbox{no more than}\ t\ \mbox{points under the dividing line}\ \ell_0)\\
&=\int_{0}^{\frac{1}{2}}\int_{0}^{\frac{1}{2}}\sum_{i=0}^{t-1}\binom{n-1}{i}(4p_1p_2)^{i}(1-4p_1p_2)^{n-1-i}\dd p_1\dd p_2\\
&=\frac{1}{4}\cdot\int_{0}^{\frac{1}{2}}\int_{0}^{\frac{1}{2}}\sum_{i=0}^{t-1}\binom{n-1}{i}(2p_1\cdot 2p_2)^{i}(1-2p_1\cdot 2p_2)^{n-1-i}\dd (2p_1)\dd (2p_2)\\
&=\frac{1}{4}\cdot\sum_{i=0}^{t-1}\binom{n-1}{i}\int_{0}^{1}\int_{0}^{1}(p_1 p_2)^{i}(1-p_1 p_2)^{n-1-i}\dd p_1\dd p_2
\end{align*}
By Lemma~\ref{lem:integration}, as $n\rightarrow\infty$,
\[\int_{0}^{1}\int_{0}^{1}(p_1 p_2)^{i}(1-p_1 p_2)^{n-1-i}\dd p_1\dd p_2=\frac{i!\log{n}}{n^{i+1}}+O\left(\frac{1}{n^{i+1}}\right).\]
Therefore
\begin{align*}
\Pr(p\in U_{[t]}) & \geq \frac{1}{4}\cdot\sum_{i=0}^{t-1}\left[\frac{(n-1)!}{n^i(n-i-1)!}\frac{\log{n}}{n} + O\left(\frac{1}{i!n}\right)\right]\\
&=\frac{1}{4}\cdot\sum_{i=0}^{t-1}\left[\left(1-\frac{1}{n}\right)\cdot\left(1-\frac{2}{n}\right)\cdots\left(1-\frac{i}{n}\right)\cdot\frac{\log{n}}{n}\cdot + O\left(\frac{1}{i! n}\right)\right]\\
&\geq \frac{1}{4}\cdot\sum_{i=0}^{t-1} \left[\left(1-\frac{(i+1)i}{2n}\right)\cdot\frac{\log{n}}{n} + O\left(\frac{1}{i! n}\right)\right]\\
&\geq \frac{1}{4}\cdot\sum_{i=0}^{t-1}\left[\left(1-\frac{(t-1)t}{2n}\right)\cdot\frac{\log{n}}{n} + O\left(\frac{1}{i!n}\right)\right]\\
&\geq\frac{1}{4}\cdot\frac{t\log{n}}{n}\cdot\left(1-\frac{t^2}{2n}\right) + O\left(\frac{1}{n}\right).\end{align*}
Finally, the expected number of points on the first $t$ Tukey layers
\[
\mathbb{E}\left|U_{[t]}\right|=\sum_{p\in X}\Pr(p\in U_{[t]})\geq \frac{1}{4}\left(1-\frac{t^2}{2n}\right)t\log{n} + O(1).
\]
The conclusions follow immediately.
\end{proof}

% !TeX root = main.tex
\section{Applications}
\label{section: application}

In this section, we discuss how our results in Sections~\ref{section: upper bound for Tukey layer} and~\ref{section: upper bound for convex layer} help in the average case analysis of two partial enclosing problems. The objective is to enclose $(n-t)$ of the given $n$ points in $X$ by a specified shape such that the area of the shape is minimized. This kind of problem is known as partial shape fitting and is an important problem in computational geometry, see, e.g., \cite{atanassov2009algorithms,das2005smallest,har2004shape,segal1998enclosing}. The points that are not enclosed are referred to as outliers \cite{atanassov2009algorithms,har2004shape}.

The average case complexity is another important measure in addition to the worst case complexity. As pointed out in \cite{dwyer1988average}, the average case analysis is desirable because the best-case
and worst-case performance of an algorithm usually differs greatly, especially for output-sensitive algorithms. In such situation, the average case complexity seems to be a more accurate and fair measurement of an algorithm's performance. A common scenario is that the input point set is drawn from some probability distribution and it is widely adopted by the computational geometry community to consider the uniform distribution in a convex polygon  \cite{affentranger1991convex,barany1999sylvester,buchta2012boundary,dwyer1988convex,efron1965convex,har2011expected}.

%When outliers are present in a point set, their geometric distribution is usually not uniform in a convex polygon. \yisays{where do the outliers come from? say something about the background of the problem. why do we care about outliers in finding convex hulls?} However, we need to treat the point set as random and assume it follows an underlying probability distribution in order to study the average case complexity. \yisays{then why do you study the average case complexity?}   The most common one in such studies is the uniform distribution from a convex polygon, which has been widely adopted in the past studies such as \cite{}.

\subsection{Enclosing Parallelogram with Minimum Area}
The algorithm given in  \cite{guo2020minimum} studies how to find a parallelogram with the minimum-area that encloses $(n-t)$ of the $n$ given points. The time complexity of the algorithm is $O\left(t^3\tau^2+n^2\log{n}\right)$, where $\tau$ is the number of points whose Tukey depth is at most $(t+1)$. Such points coincide with $U_{[t+1]}(X)$ and so $\tau=\left|U_{[t+1]}(X)\right|$. In the worst case, $\left|U_{[t+1]}(X)\right|= n$ can be true and the worst case time complexity is thus $O\left(n^2t^3+n^2\log{n}\right)$. However, on average, we have
\begin{align*}
	&\quad\, \E\left[O\left(\left|U_{[t+1]}(X)\right|^2t^3+\left|U_{[t+1]}(X)\right|tn+n^2\log{n}\right)\right]\\
	&\leq \E\left[O\left(nt^3\left|U_{[t]}(X)\right|+nt\left|U_{[t]}(X)\right|+n^2\log{n}\right)\right]\\
	&=  O\left(kt^4n\log{\frac{n}{k}}+n^2\log{n}\right),
\end{align*}
when $X$ is uniformly sampled from a $k$-gon. When $t$ is between $\Omega\left(\log^{\frac{1}{3}}{n}\right)$ and $ O\left(\frac{n}{k\log{\frac{n}{k}}}\right)$, the average case complexity is smaller than the worst-case complexity. This explains why in many cases the actual runtime of the algorithm is faster than the worst-case complexity.

\subsection{Minimum Enclosing Convex Hull}
\label{subsec:Minimum Enclosing Convex Hull}
Another application of our result is the algorithm for the minimum enclosing convex hull. Let $X$ be a set of $n$ points in $\R^2$. The problem asks to find a subset $X' \subset X$, $|X'| =t$, such that area  of $H_t(X \setminus X')$ is minimized.  In \cite{atanassov2009algorithms}, Atanassov et al.\ provide an elegant solution to this problem with running time $O \left(    n \log n  + \binom{4t}{2t} (3t)^t |H_{[t]} (X)|   \right)$. In the worst case, $|H_{[t]}(X)| = n$, which happens when $X$ has  at most $t$ layers. For the average case, Theorem~\ref{thm:upper bound 2} implies a time complexity of $O\left(  n \log n + k \binom{4t}{2t} (3t)^{t} t^3 \log \frac{n}{kt^2} \right)$, when $X$ is uniformly
distributed in convex $k$-gon. The average case is substantially better than the worst case when $t = O\left(\left(\frac{n}{k\log(n/k)})\right)^{1/3} \right)$.

%!TeX root = main.tex
\section{Closing Remarks}
In this paper, we studied the expected size of the random convex layers and random Tukey layers of a point set $X$ consisting of $n$ points drawn independently and uniformly from a convex $k$-gon.

For random Tukey layers, we showed that $\E|U_{[t]}(X)| = O\left(kt \log(n/k)\right)$ but only showed a matching lower bound of $\Omega(t\log n)$ for triangles and parallelograms. We leave an open problem of obtaining a general lower bound of $\Omega(kt\log n)$, for which a straightforward extension of our current technique of considering a line passing through a single point $p$ in Section~\ref{section: low bound for convex shell} seems inadequate. We also leave an open problem of obtaining a tight constant in the asymptotic results (which could depend on $t$); our constants are $4$ in the upper bound and $2$ in the lower bound, which are not tight since the tight constant is known to be $8/3$ when $t=1$~\cite{RS63}.

For random convex layers, we showed that $\E|V_{[t]}(X)| = O(kt^3\log(n/(kt^2)))$. However, when the points are from sampled from a square, a better upper bound of $O(t^2\log(n/t^2))$ is known~\cite{he2018maximal}. Thus, a natural question is whether it holds $\E|V_{[t]}(X)| = O(kt^2\log(n/(kt^2)))$ in general. Another interesting open problem is to obtain a lower bound with dependence on $t$, as existing lower bounds are only for $t=1$ and there seem substantial difficulties to extend the existing techniques to a larger $t$.

\bibliographystyle{elsarticle-num}
\bibliography{reference}

\end{document}